\documentclass[11pt]{amsart}
\usepackage[utf8]{inputenc}
\usepackage{amsmath,amsfonts,amssymb,amsthm}
\usepackage{graphicx,amsmath, amscd, hyperref,amsthm}
\usepackage{color,comment}
\usepackage[lmargin=3cm,rmargin=3cm,tmargin=3cm,bmargin=3cm]{geometry}
\usepackage{accents}

\newtheorem{theorem}{Theorem}
\newtheorem{lemma}{Lemma}
\newtheorem{assumption}{Assumption}
\newtheorem{example}{Example}

\newtheorem*{main-result*}{Main result}
\newcommand{\W}{W}

\newcommand{\E}{\mathbb{E}}
\newcommand{\N}{\mathbb{N}}

\newcommand{\Cov}{{\rm Cov}}

\newcommand{\var}{{\rm var}}
\newcommand{\tr}{{\rm Tr}}

\newcommand{\overbar}[1]{\mkern 1.5mu\overline{\mkern-1.5mu#1\mkern-1.5mu}\mkern 1.5mu}

\newlength{\dhatheight}

\title{A note on quantum expanders}

\author{C\'ecilia Lancien} 
\author{Pierre Youssef}

\address{C\'ecilia Lancien. CNRS \& Institut Fourier, Universit\'e Grenoble Alpes, 118 rue des maths, 38610 Gi\`eres, France.}
\email{cecilia.lancien@univ-grenoble-alpes.fr}
\address{Pierre Youssef. Division of Science, NYU Abu Dhabi, Saadiyat Island, Abu Dhabi, UAE \& Courant Institute of Mathematical Sciences, New York University, 251 Mercer st, New York, NY 10012, USA.}
\email{yp27@nyu.edu}

\date{June 13 2025}


\begin{document}
	
\begin{abstract}
We prove that a wide class of random quantum channels with few Kraus operators, sampled as random matrices with some sparsity and moment assumptions, typically exhibit a large spectral gap, and are therefore optimal quantum expanders. In particular, our result provides a recipe to construct random quantum expanders from their classical (random or deterministic) counterparts. This considerably enlarges the list of known constructions of optimal quantum expanders, which was previously limited to few examples. Our proofs rely on recent progress in the study of the operator norm of random matrices with dependence and non-homogeneity, which we expect to have further applications in several areas of quantum information. 
\end{abstract}

\maketitle
	
\section{Introduction}

\subsection{Quantum states and channels} \hfill\vspace{0.1cm}

In what follows, for every $m\in \mathbb{N}$, we denote by $\mathcal M_m(\mathbb C)$ the set of $m\times m$ matrices with complex entries. Given $X\in \mathcal M_m(\mathbb C)$, we denote by $X^*$ its adjoint, by $X^t$ its transpose and by $\overbar{X}$ its (entry-wise) conjugate, and we write $X\succeq 0$ when $X$ is self-adjoint positive semidefinite. 

In quantum physics, the state of an $n$-dimensional system is described by a density operator on $\mathbb C^n$, i.e.~$\rho\in\mathcal M_n(\mathbb C)$ a self-adjoint positive semidefinite matrix with trace $1$. A transformation of such quantum system is described by a quantum channel, i.e.~$\Phi:\mathcal M_n(\mathbb C)\to\mathcal M_n(\mathbb C)$ a completely positive and trace-preserving linear map. Recall that a linear map $\Phi:\mathcal M_n(\mathbb C)\to\mathcal M_n(\mathbb C)$ is said to be
 \begin{itemize}
	\item positive if, for all $X\in \mathcal M_n(\mathbb C)$, $X\succeq 0 \Rightarrow \Phi(X)\succeq 0$, and completely positive (CP) if $\Phi\otimes I:\mathcal M_{n^2}(\mathbb C)\to\mathcal M_{n^2}(\mathbb C)$ is positive (where $I:\mathcal M_n(\mathbb C)\to\mathcal M_n(\mathbb C)$ denotes the identity map),
	\item trace-preserving (TP) if, for all $X\in \mathcal M_n(\mathbb C)$, $\tr(\Phi(X))=\tr(X)$.
\end{itemize}
  	
The action of a CP map $\Phi$ on $\mathcal M_n(\mathbb C)$ can always be described in the following (non-unique) way, called a Kraus representation (see e.g.~\cite[Section 2.3.2]{aubrun2017} or \cite[Chapter 2]{wolf2012}):
\begin{equation} \label{eq:Kraus}
	\Phi:X\in\mathcal M_n(\mathbb C) \mapsto \sum_{s=1}^d K_sXK_s^*\in\mathcal M_n(\mathbb C), 
\end{equation}
for some $d\in\N$ and some $K_1,\ldots,K_d\in\mathcal M_n(\mathbb C)$, called Kraus operators of $\Phi$. The fact that $\Phi$ is TP is equivalent to the following constraint on the $K_s$'s:
\begin{equation} \label{eq:TP-condition} \sum_{s=1}^d K_s^*K_s = I. \end{equation}
The smallest $d$ such that an expression of the form of equation \eqref{eq:Kraus} for $\Phi$ exists is called the Kraus rank of $\Phi$.

An equivalent way of characterizing the Kraus rank of a CP map $\Phi:\mathcal M_n(\mathbb C)\to\mathcal M_n(\mathbb C)$ is as the rank (i.e.~the number of non-zero eigenvalues) of its associated Choi matrix $C_\Phi\in\mathcal M_{n^2}(\mathbb C)$ (see again \cite[Section 2.3.2]{aubrun2017} or \cite[Chapter 2]{wolf2012}). The latter is defined as
\[ C_\Phi:= \Phi\otimes I\left(\sum_{i,j=1}^n E_{ij}\otimes E_{ij}\right) = \sum_{i,j=1}^n \Phi(E_{ij})\otimes E_{ij}, \]
where $E_{ij}\in\mathcal M_n(\mathbb C)$ has its entry $(i,j)$ equal to $1$ and all other entries equal to $0$. As a consequence of this alternative definition, it is clear that the Kraus rank of a CP map on $\mathcal M_n(\mathbb C)$ is always at most $n^2$ (as is the rank of a matrix on $\mathbb C^{n^2}$).
	
By the analogue of Perron-Frobenius theory to this context (see e.g.~\cite[Chapter 6]{wolf2012}), a quantum channel $\Phi$ on $\mathcal M_n(\mathbb C)$ always has a largest (in modulus) eigenvalue $\lambda_1(\Phi)$ which is equal to $1$, and consequently a largest singular value $s_1(\Phi)$ which is at least $1$. In addition, this largest eigenvalue has an associated eigenvector which is a positive semidefinite matrix. Hence $\Phi$ always has a fixed state, i.e.~a quantum state $\hat\rho$ such that $\Phi(\hat\rho)=\hat\rho$. $\Phi$ is said to be unital if its fixed state is the so-called maximally mixed state, i.e.~$\hat\rho=I/n$. It turns out that, if $\Phi$ is unital, then its largest singular value is also equal to $1$ (see e.g.~\cite[Chapter 4]{watrous2018}).

The constraint of being unital is dual to that of being TP. What we mean is that a CP map $\Phi$ is unital if and only if its dual (or adjoint) CP map $\Phi^*$, for the Hilbert-Schmidt inner product, is TP. It thus reads at the level of Kraus operators as 
 \begin{equation} \label{eq:unital-condition} \sum_{s=1}^d K_sK_s^* = I. \end{equation}
 
Note that, identifying $\mathcal M_n(\mathbb C)$ with $\mathbb C^n\otimes\mathbb C^n$, a linear map $\Phi:\mathcal M_n(\mathbb C)\to\mathcal M_n(\mathbb C)$ can equivalently be seen as a linear map $M_\Phi:\mathbb C^n\otimes\mathbb C^n\to\mathbb C^n\otimes\mathbb C^n$, i.e.~an element of $\mathcal M_{n^2}(\mathbb C)$. Concretely, a CP linear map
\[ \Phi:X\in\mathcal M_n(\mathbb C) \mapsto \sum_{s=1}^d K_sXK_s^*\in\mathcal M_n(\mathbb C) \]
can be identified with
\[ M_\Phi = \sum_{s=1}^d K_s\otimes\overbar K_s \in \mathcal M_{n^2}(\mathbb C). \]
This identification preserves the eigenvalues and singular values, i.e.~$\lambda_k(\Phi)=\lambda_k(M_\Phi)$ and $s_k(\Phi)=s_k(M_\Phi)$ for each $1\leq k\leq n^2$.

\subsection{Quantum expanders} \label{sec:def-expanders} \hfill\vspace{0.1cm}
 
Let $\mathbf\Phi:=(\Phi_{n,d_n})_{n\in\mathbb N}$ be a sequence of quantum channels, where for each $n\in\mathbb N$, $\Phi_{n,d_n}$ acts on $\mathcal M_n(\mathbb C)$ and has Kraus rank $d_n$. $\mathbf\Phi$ is called an expander if it satisfies the following properties, the third one being actually less crucial, as we will explain below:
\begin{itemize}
	\item[(1)] The $\Phi_{n,d_n}$'s have asymptotically a small Kraus rank, i.e.
 \[ \frac{d_n}{n^2}\underset{n\to\infty}{\longrightarrow} 0 . \]
	\item[(2)] The $\Phi_{n,d_n}$'s have asymptotically a small second largest (in modulus) eigenvalue, i.e.\ there exist $0\leq\varepsilon<1$ and $n_0\in\mathbb N$ such that, for all $n\geq n_0$,
 \[ \left|\lambda_2\left(\Phi_{n,d_n}\right)\right| \leq \varepsilon . \] 
        \item[(2')] The $\Phi_{n,d_n}$'s have asymptotically a small second largest singular value, i.e.\ there exist $0\leq\varepsilon<1$ and $n_0\in\mathbb N$ such that, for all $n\geq n_0$,
 \[ s_2\left(\Phi_{n,d_n}\right) \leq \varepsilon . \] 
        \item[(3)] The $\Phi_{n,d_n}$'s have asymptotically a fixed state with large entropy, i.e.\ denoting by $\hat\rho_{n,d_n}$ the fixed state of $\Phi_{n,d_n}$,
\[ S\left(\hat\rho_{n,d_n}\right) \underset{n\to\infty}{\longrightarrow} \infty . \]
\end{itemize}

In condition (3) above, $S(\cdot)$ stands for the von Neumann entropy, defined, for any state $\rho$, as $S(\rho):=-\tr(\rho\log\rho)$. Note that, given a state $\rho$ on $\mathbb C^n$, we always have $0\leq S(\rho)\leq\log n$, with equality in the first inequality iff $\rho$ is a pure state (i.e.~a rank-$1$ projector) and equality in the second inequality iff $\rho$ is the maximally mixed state. The entropy of a quantum state can thus be seen as a measure of how mixed it is: if $S(\rho)\geq\log r$ for some $1\leq r\leq n$, then the rank of $\rho$ is at least $r$. The advantage of the entropy over the rank is that it is a smoother quantity (which captures how uniformly distributed the non-zero eigenvalues of $\rho$ are rather than just how many they are). 

Let us explain the meaning of conditions (2) and (2'). Given a quantum channel $\Phi$ with fixed state $\hat\rho$, its second largest eigenvalue or singular value can be seen as quantifying how far $\Phi$ is from the ``ideal'' quantum channel $\Phi_{\hat\rho}:X\mapsto(\tr X)\hat\rho$, which sends any input state on $\hat\rho$. Indeed, $|\lambda_2(\Phi)|=|\lambda_1(\Phi-\Phi_{\hat\rho})|$ and $s_2(\Phi)=s_1(\Phi-\Phi_{\hat\rho})$, so the smaller $|\lambda_2(\Phi)|$ or $s_2(\Phi)$, the closer, in a sense, $\Phi$ to $\Phi_{\hat\rho}$. Now, the fact that $\hat\rho$ has a large entropy implies that $\Phi_{\hat\rho}$ has a large Kraus rank. More precisely, if $S(\hat\rho)\geq\alpha\log n$ for some $0\leq\alpha\leq 1$, then $\hat\rho$ has rank at least $n^\alpha$, and therefore $\Phi_{\hat\rho}$ has Kraus rank at least $n^{1+\alpha}$. The latter claim can be checked by observing that the Choi matrix associated to $\Phi_{\hat\rho}$ is simply $C_{\Phi_{\hat\rho}} = \hat\rho\otimes I$, which indeed has rank $n^{1+\alpha}$. 
Consequently, the fact that $\Phi$ has a small Kraus rank means that it is an approximation of $\Phi_{\hat\rho}$ which is much more economical than $\Phi_{\hat\rho}$ itself. This is what makes it particularly useful in practice. We see from this discussion that what we actually need for the quantum channel $\Phi$ to provide a significant Kraus rank reduction compared to $\Phi_{\hat\rho}$ is that it satisfies a joint condition (1) and (3) which would be: $d\ll n^{1+\alpha}$, where $\alpha=S(\hat\rho)/\log n$. In particular, if $d\ll n$, then $\Phi$ fulfils the latter condition whatever the entropy of $\hat\rho$. 

If the quantum channel $\Phi$ is self-adjoint, in the sense that $\Phi^*=\Phi$, then conditions (2) and (2') above are equivalent, because $|\lambda_2(\Phi)|=s_2(\Phi)$. If it is unital, then condition (2') is stronger than condition (2), because by Weyl's majorant theorem (see e.g.~\cite[Theorem~II.3.6]{MR1477662}) $|\lambda_1(\Phi)|+|\lambda_2(\Phi)| \leq s_1(\Phi)+s_2(\Phi)$, so the fact that $|\lambda_1(\Phi)|=s_1(\Phi)=1$ implies that $|\lambda_2(\Phi)| \leq s_2(\Phi)$. In general however, these two conditions are not comparable to one another. And depending on the context, either one or the other might be more relevant. For instance, $|\lambda_2(\Phi)|$ quantifies the speed of convergence of the dynamics $(\Phi^N(\rho))_{N\in\mathbb N}$ to its equilibrium $\hat\rho$ and the speed of decay of correlations in the 1D many-body quantum state that has $\Phi$ as so-called transfer operator. On the other hand, $s_2(\Phi)$ generally appears when quantifying the distance of $\Phi$ to $\Phi_{\hat\rho}$ in certain norm distances: it is exactly equal to the $2{\to}2$ norm (and can be related to several other $p{\to}q$ norms) of $\Phi-\Phi_{\hat\rho}$. In the remainder of this paper, we will mostly use condition (2') rather than condition (2) as our definition of expansion (i.e.~study expansion in terms of second largest singular value rather than second largest eigenvalue).

Let us briefly comment here on the fact that expander quantum channels can be seen as an analogue of classical expander graphs (see e.g.~\cite{hoory2006} for a complete review on the topic). An expander graph (say in the undirected case, for simplicity) is usually defined as a regular graph $G$ (i.e.~whose vertices all have the same degree) that combines the two properties of being sparse (i.e.~the degree is small compared to the number of vertices) and of having a normalized adjacency matrix $A$ (which is symmetric in this case) that has a large upper spectral gap (i.e.~$|\lambda_2(A)|$ is small compared to $1$). Since $G$ is regular, we have that $A$ leaves the uniform probability vector $u$ (which has maximal entropy) invariant. And the second largest eigenvalue of $A$ measures the distance of $A$ to $A_u$, the transition matrix that sends any input probability vector on $u$, as $|\lambda_2(A)|=|\lambda_1(A-A_u)|$. Now, $A_u$ is nothing else than the normalized adjacency matrix of the complete graph. So an expander can be seen as a graph which approximates, in a sense, the complete graph, even though it is sparse. To transpose this to the quantum setting, probability vectors are replaced by density operators, transition matrices are replaced by quantum channels, and the corresponding definition of an expander indeed becomes the one that we have given earlier.

In both the classical and quantum settings, the definition of expansion that we gave here is a spectral one. The latter is closely related to a more geometrical characterization of expansion, that quantifies either how uniformly distributed the edges of the graph are or how uniformly spread the output directions of the quantum channel are. Both classically and quantumly, spectral and geometrical expansion are known to be somewhat equivalent notions, in the sense that there are two-way inequalities relating the parameters measuring one or the other \cite{MR875835,bannink2020}. We work with spectral expansion here, as people often do, as it is mathematically more tractable.

\subsection{Previously known results and our contribution} \hfill\vspace{0.1cm}
	
Recently, several attempts have been made at exhibiting examples of such quantum expanders. It is known that any quantum channel $\Phi$ with Kraus rank $d$ verifies  $|\lambda_2(\Phi)|\geq c/\sqrt{d}$ and $s_2(\Phi)\geq c'/\sqrt{d}$, where $c,c'>0$ are universal constants (i.e.~independent of the underlying dimension $n$) but whose optimal values are known only for specific classes of quantum channels, such as unital ones \cite{hastings2007,BenAroya2008, timhadjelt2024}. Note that this is in complete analogy with the classical case and the corresponding Alon-Boppana lower bound on the second largest eigenvalue of the normalized adjacency matrix of a $d$-regular graph, namely $2\sqrt{d-1}/d-o(1)$ for large $n$, which is known to be tight \cite{MR875835,MR1124768}. 
An ``optimal'' quantum expander can thus be defined as a quantum channel $\Phi_{n,d}$ on $\mathcal M_n(\mathbb C)$ having Kraus rank $d$ and satisfying either $|\lambda_2(\Phi_{n,d})|\leq C/\sqrt{d}$ or $s_2(\Phi_{n,d})\leq C'/\sqrt{d}$ for some universal constants $C,C'>0$ (i.e.~whose second largest eigenvalue or singular value has the optimal scaling in $d$). In a similar fashion, a $d$-regular graph is called a Ramanujan graph when it is an (exactly) optimal classical expander, i.e.~when the second largest eigenvalue of its normalized adjacency matrix is upper bounded by $2\sqrt{d-1}/d$.  

In the classical setting, explicit constructions of families of Ramanujan graphs were achieved by Margulis \cite{MR939574} and Lubotzky, Philips and Sarnak \cite{MR963118} for specific values of $d$. The existence of Ramanujan $d$-regular bipartite graphs for every $d\geq 3$ was proved in a breakthrough work of Marcus, Spielman and Srivastava \cite{MR3374962}. Moreover, several models of random graphs have been shown to be almost Ramanujan. For instance, random $d$-regular graphs with fixed degree $d\geq 3$ \cite{MR2437174,MR4203039}, random $d$-regular graphs with growing degree $d$ \cite{MR1634360,MR3758727,MR3909972,MR4135670,Sarid}, and (moderately) dense Erd\H{o}s-R\'{e}nyi graphs, i.e.~with parameter $d\gg\log n$ \cite{MR3878726, MR3945756,MR4234995}.

The first attempts at exhibiting quantum expanders were inspired by classical constructions, in particular those based on Cayley graphs. This provided (semi-)explicit examples of quantum expanders, but none of them was optimal, either because the classical one was not (for an Abelian underlying group \cite{ambainis2004}) or because something was lost in the classical-to-quantum embedding (for a non-Abelian underlying group \cite{BenAroya2008,harrow2008,eisert2008}). In fact, there are only three known constructions of optimal quantum expanders up to date, all of which being random. The first model, considered by Hastings \cite{hastings2007}, consisted in sampling Kraus operators as $d/2$ independent Haar distributed unitaries on $\mathbb C^n$ and their adjoints, so that the corresponding random CP map was self-adjoint. This model was later analyzed differently by Pisier \cite{pisier2014} and very recently by Timhadjelt \cite{timhadjelt2024}, who also showed that its non-self-adjoint version exhibited a similar behaviour (i.e.~simply sampling Kraus operators as $d$ independent Haar distributed unitaries on $\mathbb C^n$). Then Gonz\'alez-Guill\'en, Junge and Nechita \cite{gonzalez2018} studied the case where Kraus operators are blocks of a Haar distributed isometry from $\mathbb C^n$ to $\mathbb C^n\otimes\mathbb C^d$. Finally, Lancien and P\'erez-Garc\'ia \cite{lancien2022} looked at a model where Kraus operators are sampled as $d$ independent Gaussian matrices on $\mathbb C^n$. The latter work is the only one where expansion is proved for a Kraus rank $d$ that is either fixed or growing with the dimension $n$: all other works studied only the regime of fixed $d$.
	
The goal of this work is to show that many other models of random quantum channels, beyond the three mentioned above, provide examples of optimal expanders as well. We will see that, in the regime of a Kraus rank that grows at least poly-logarithmically with the dimension, quantum expanders can be constructed by taking Kraus operators to be random matrices with independent entries that are quite arbitrary, as long as they satisfy some moments' growth assumption. Additionally, the model we introduce is general enough so that it incapsulates random matrices with a variance profile and even sparse random matrices (hence being a first step towards derandomization of optimal constructions). 
Informally, the recipe to construct optimal quantum expanders with Kraus rank at most $d$ goes as follows:
\begin{enumerate}
    \item Take $P$ to be the transition matrix of an irreducible Markov chain on $n$ sites where no transition probability is larger than some power of $1/\log n$, and with second singular value of order at most $1/\sqrt{d}$. As illustrative examples, $P$ could simply be the matrix with all entries equal to $1/n$, or more interestingly the adjacency matrix of an almost Ramanujan $d$-regular graph on $n$ vertices (with the extra assumption that $d$ is at least some power of $\log n$). 
    \item Take a random $n\times n$ complex matrix $W$ with independent centered entries satisfying some moment growth condition, and with variance profile given by $P^t$. 
    \item Take $W_1,\ldots,W_d$ independent copies of $W$ and define the random CP map $\Phi$ on $\mathcal M_n(\mathbb C)$ whose matrix representation is 
    $$M_\Phi=\frac{1}{d}\sum_{s=1}^d \W_s\otimes\overbar\W_s. $$ 
    Then $\Phi$ satisfies on average the TP condition \eqref{eq:TP-condition} and is with high probability an optimal quantum expander, as we prove in this note. 
\end{enumerate}

As already mentioned, this general procedure illustrates the abundance of examples of quantum expanders which could be realized. In particular, it is worth highlighting that the variance profile $P^t$ of the random Kraus operators can potentially have only $d$ non-zero entries per line and column (hence being very sparse for $d\ll n$), as soon as those non-zero entries correspond to edges of an optimally expanding $d$-regular graph. In this case, the above recipe uses an optimal classical expander (either random or deterministic) in order to construct a (random) optimal quantum expander. 
The exact moment assumptions as well as the conditions on the variance profile are stated in Section~\ref{sec:model} (Assumptions~\ref{a:indep entries}, \ref{a: profil variance}, \ref{a: 4 moment}, \ref{a: max variance}) and Section~\ref{sec:application} (Assumption~\ref{a:eta}). 
To keep the introduction light, we will state our main result in an informal way here and differ the precise definition of the class of models as well as corresponding examples to the next sections.  

 \begin{main-result*}[Informal]
 Given $n,d\in\mathbb N$ such that $(\log n)^{12}\leq d\leq n^2$, let $W\in\mathcal M_n(\mathbb C)$ be a random matrix with independent centered entries whose variance profile $\eta\in\mathcal M_n(\mathbb R)$ is a stochastic matrix such that $s_2(\eta)\leq C_0/\sqrt{d}$, and let $W_1,\ldots,W_d$ be independent copies of $W$. Define the random CP map $\Phi$ on $\mathcal M_n(\mathbb C)$, having Kraus rank at most $d$, as
	$$\Phi:X\in\mathcal M_n(\mathbb C)\mapsto \frac{1}{d}\sum_{s=1}^d W_sXW_s^*\in\mathcal M_n(\mathbb C). $$
	Then, $\Phi$ is such that, on average, its Kraus operators satisfy the TP condition \eqref{eq:TP-condition} and, with high probability, 
 $$ s_1(\Phi)\geq 1-\frac{C}{\sqrt{d}} \quad \text{and} \quad s_2(\Phi)\leq \frac{C}{\sqrt{d}}. $$ 
 What is more, by suitably renormalizing $W_1,\ldots,W_d$, one can obtain a corresponding random CP map $\tilde\Phi$ on $\mathcal M_n(\mathbb C)$, having Kraus rank at most $d$, which is TP (i.e.~a quantum channel) and such that, with high probability, $s_2(\tilde\Phi)\leq \tilde C/\sqrt{d}$.
\end{main-result*}

The rigorous statement corresponding to the previous informal statement, as well as to the description preceding it, is stated as Theorems~\ref{th:expander} and \ref{th:expander-TP}. 
The proof of this result crucially relies on recent progress in the understanding of the operator norm of non-homogeneous random matrices, and more precisely on refinements of the noncommutative Khintchine inequality (see \cite{junge2013,MR3531673,MR3878726,vanHandel,vanHandel2} and references therein). 
One of the motivations of the current work, beyond the specific result about quantum expanders, is in fact to bring these powerful results to the attention of the quantum information community. We expect them to find further applications in view of their generality.

 \subsection{Organization of the paper and notation} \hfill\vspace{0.1cm}

 The remainder of this paper is organized as follows. In Section \ref{sec:technical} we introduce a very general random matrix model with tensor product structure and prove an upper bound on its operator norm. The key technical tool that we use is a recent powerful result from \cite{vanHandel,vanHandel2}, which allows to estimate the operator norm of random matrices whose entries have dependence and non-homogeneity. In Section \ref{sec:application} we then show how the main result from Section \ref{sec:technical} implies that a wide class of random quantum channels are on average expanders. Finally, in Section \ref{sec:conclusion}, we discuss some consequences and perspectives.

 In what follows, given $x\in\mathbb C^m$, we denote by $\|x\|$ its Euclidean norm. And given $X\in\mathcal M_m(\mathbb C)$, we denote by $\|X\|_p:=(\tr(|X|^p))^{1/p}$ its Schatten $p$-norm, for all $p\geq 1$. In particular, $\|X\|_\infty$ stands for its operator (or spectral) norm (i.e.~its largest singular value) and $\|X\|_2$ stands for its Hilbert-Schmidt norm. Also, given a linear map $T$ (either on $\mathbb C^m$ or on $\mathcal M_m(\mathbb C)$), we denote by $\lambda_1(T),\lambda_2(T),\ldots$ its eigenvalues, with modulus in non-increasing order, and by $s_1(T),s_2(T),\ldots$ its singular values, in non-increasing order. 
	
\section{Operator norm estimate for a random matrix model with tensor product structure} \label{sec:technical} 

\subsection{Random matrix model} \label{sec:model} \hfill\vspace{0.1cm}
	
Let $\W=(W_{ij})_{1\leq i,j\leq n}$ be an $n\times n$ random matrix (either real or complex) satisfying the following assumptions:
	
	\begin{assumption}\label{a:indep entries}
		The entries $W_{ij}$, $1\leq i,j\leq n$, are independent, of mean $0$ and of variance $\eta_{ij}$, i.e.
  $$\forall\ 1\leq i,j\leq n,\ \E\left(W_{ij}\right)=0 \ \text{ and } \ \E\left(\left|W_{ij}\right|^2\right)=\eta_{ij}. $$ 
	\end{assumption}
	\begin{assumption}\label{a: profil variance}
		The variance matrix $\eta=(\eta_{ij})_{1\leq i,j\leq n}$ is doubly sub-stochastic, i.e.~each of its rows and columns sums up to at most $1$: 
		$$
		\forall \ 1\leq i\leq n,\ \sum_{j=1}^n \eta_{ij}\leq 1   \quad \text{and} \quad  \forall \ 1\leq j\leq n,\ \sum_{i=1}^n \eta_{ij}\leq 1.
		$$
	\end{assumption}
 	\begin{assumption}\label{a: 4 moment}
		There exist a universal constant $\mathbf{C}>0$ and a parameter $\beta\geq 0$ such that, for every $p\in \N$,
		$$
		\forall\ 1\leq i,j\leq n,\ \left(\E\left(\left|W_{ij}\right|^{2p}\right)\right)^{1/p} \leq \mathbf{C} p^\beta \eta_{ij}.
		$$
	\end{assumption}
	\begin{assumption}\label{a: max variance}
		There exists a universal constant $\mathbf{C'}>0$ such that 
		$$ \max_{1\leq i, j\leq n} \eta_{ij}\leq \frac{\mathbf{C'}}{(\log n)^{\alpha_\beta}}, $$
  where $\alpha_\beta=\max(\beta,2)$.
	\end{assumption}

	The class of random matrices satisfying the above assumptions is quite large. We provide some examples below. 
	
	\begin{example}[i.i.d.~heavy-tailed entries] \label{ex:bounded-moments}
		Consider a random variable $\xi$ (either real or complex) of mean $0$, variance $1$ and satisfying, for some universal constant $C>0$ and some parameter $\beta\geq 0$, the moment growth condition 
  $$\forall\ p\in\mathbb N, \ \left(\E\left(|\xi|^{2p}\right)\right)^{1/p} \leq C p^{\beta}.
  $$
  Special cases include $\beta=1$, where $\xi$ is called sub-Gaussian, and $\beta=2$, where $\xi$ is called sub-exponential. But we also allow for larger $\beta$, i.e.~for $\xi$ having heavier tails. 
  Consider an $n\times n$ random matrix $M$ whose entries are i.i.d.~copies of $\xi$, and set 
		$\W=M/\sqrt{n}$. 
		Clearly, $W$ satisfies all required assumptions, provided $n$ is large enough. 
	\end{example}
	
	\begin{example}[Gaussian matrices with variance profile] \label{ex:gaussian}
		Consider an $n\times n$ random matrix $W$ whose entries are independent and, for every $1\leq i,j\leq n$, $W_{ij}$ is a Gaussian random variable (either real or complex) with mean $0$ and variance $\sigma_{ij}^2$. Assume that the variance profile matrix $\sigma=(\sigma_{ij}^2)_{1\leq i,j\leq n}$ is doubly stochastic (i.e.~each of its rows and columns sums up to $1$) and satisfies 
		$$
		\max_{1\leq i,j\leq n} \sigma_{ij}^2\leq \frac{1}{(\log n)^2}. 
		$$
		Then $W$ satisfies all required assumptions, with $\beta=1$. 
	\end{example}

	\begin{example}[Random matrices with deterministic sparsity pattern] \label{ex:sparse}
		Let $n,r\in \N$ with $r\geq (\log n)^2$. Consider any undirected $r$-regular graph on $n$ vertices and let $A=(A_{ij})_{1\leq i,j\leq n}$ be the corresponding adjacency matrix. Let $\xi$ be a uniformly bounded random variable (either real or complex) of mean $0$ and variance $1$.  
		Let $(\xi_{ij})_{1\leq i,j\leq n}$ be i.i.d copies of $\xi$ and for every $1\leq i,j\leq n$, set $W_{ij}=A_{ij}\xi_{ij}/\sqrt{r}$. 
		Then $\W= (W_{ij})_{1\leq i,j\leq n}$ inherits the sparsity pattern of the undirected $r$-regular graph. 
		Moreover, since $\eta_{ij}= \eta_{ji}= A_{ij}/r$ for every $1\leq i,j\leq n$, $W$ clearly satisfies Assumptions~\ref{a:indep entries} and \ref{a: profil variance}. 
		Finally, since $\xi$ is bounded and $r\geq(\log n)^2$, one can check that $W$ satisfies Assumptions~\ref{a: 4 moment} and \ref{a: max variance}, with $\beta=0$. 
	\end{example}

Equipped with the class of random matrices introduced above, we now construct the random CP map which will be proven in the sequel to have a large gap between its first and second largest singular values. We let $\W\in\mathcal M_n(\mathbb C)$ be such random matrix, $\W_1,\ldots,W_d$ be independent copies of $W$, and we set 
\begin{equation} \label{eq:def-Y}
Y=\frac{1}{d}\sum_{s=1}^d \W_s\otimes\overbar\W_s \in\mathcal M_{n^2}(\mathbb C).
\end{equation}
In what follows, we will always implicitly assume that
$$ d\leq n^2, $$
the most interesting regime actually being when $d \ll n^2$. We will often need, for technical reasons, to suppose further that
$$ d\geq (\log n)^{\gamma}, $$
for some $\gamma>0$ that will depend on the context.
Our goal will now consist in controlling the first and second largest singular values of the random matrix $Y$ as defined by equation \eqref{eq:def-Y}, and hence of the corresponding random CP map
$$\Phi_Y:X\in\mathcal M_n(\mathbb C)\mapsto \frac{1}{d}\sum_{s=1}^d \W_sX\W_s^*\in\mathcal M_n(\mathbb C).$$
In order to do this, we will first aim at deriving an upper bound for $\Vert Y-\E(Y)\Vert_\infty$. The final result that we obtain is stated below.

\begin{theorem} \label{th:main-result}
Let $\W$ be an $n\times n$ random matrix satisfying Assumptions \ref{a:indep entries}, \ref{a: profil variance}, \ref{a: 4 moment} and \ref{a: max variance}. Let $Y=\frac{1}{d}\sum_{s=1}^d \W_s\otimes\overbar\W_s$, where $\W_s$, $1\leq s\leq d$, are independent copies of $\W$. If $d\geq (\log n)^{12+\delta}$, for some $\delta\geq 0$, then setting $\delta'=\min(\delta/12,1/4)$, we have
$$ \E \Vert Y-\E(Y) \Vert_\infty \leq \frac{2}{\sqrt{d}}\left(1+\frac{C_\beta}{(\log n)^{\delta'}}\right),$$
and what is more, 
$$ \mathbb P\left( \Vert Y-\E(Y) \Vert_\infty \leq \frac{2}{\sqrt{d}}\left(1+\frac{C_\beta'}{(\log n)^{\delta'}}\right) \right) \geq 1-\frac{2}{d^{1/4}}, $$
where $C_\beta,C_\beta'>0$ are constants depending only on $\beta$. 
\end{theorem}
The remainder of this section is devoted to the proof of this theorem. 

 \subsection{Operator norm of random matrices with dependence and non-homogeneity} \hfill\vspace{0.1cm}
 
Understanding the spectrum of a random matrix and in particular its operator norm is an extremely well studied topic. While models with independent and identically distributed entries were initially considered and well understood, there has been a lot of recent activity considering models with dependence and non-homogeneity. Since our model exhibits both some dependence (because of the tensor product structure) and non-homogeneity, we will make use of recent advances in the study of the operator norm of such random matrices. More precisely, we will rely on the following result, which is a slight modification of \cite[Corollary 2.17]{vanHandel2}, obtained by combining it with \cite[Theorem 2.8]{vanHandel2}. Indeed, \cite[Corollary 2.17]{vanHandel2} as it is stated applies to the case of random matrices that are almost surely bounded, which is quite restrictive, but truncation methods, as described in \cite[Section 8]{vanHandel2} and whose final product is summarized in \cite[Theorem 2.8]{vanHandel2}, allow to relax this assumption.

\begin{theorem}[\cite{vanHandel,vanHandel2}]\label{t: van handel}
Let $Z_1,\ldots, Z_T$ be independent and centered $m\times m$ random matrices (either real or complex) and let $X=\sum_{s=1}^T Z_s$. Let $\Cov(X)$ denote the $m^2\times m^2$ covariance matrix associated to $X$, i.e.~$\Cov(X)_{ijkl}=\E(X_{ij}\overbar{X}_{kl})$ for every $1\leq i,j,k,l\leq m$. 
And define the following parameters:
	\begin{align*}
		& \sigma(X):= \max\left(\|\E(XX^*)\|_\infty^{1/2},\|\E(X^*X)\|_\infty^{1/2}\right), \\
		& \upsilon(X):= \|\Cov(X)\|_\infty^{1/2}, \\
            & R(X) := \left(\E\left(\max_{1\leq s\leq T}\|Z_s\|_\infty^2\right)\right)^{1/2} .
	\end{align*}
    We then have, assuming that $\sigma(X)\geq C_0(\log m)^3R(X)$, 
    \begin{align*} 
	\E \|X\|_\infty & \leq \|\E(XX^*)\|_\infty^{1/2}+\|\E(X^*X)\|_\infty^{1/2} \\
    & \qquad + C\left( (\log m)^{3/4} \sigma(X)^{1/2}\upsilon(X)^{1/2} + (\log m)\sigma(X)^{5/6}R(X)^{1/6} \right), 
    \end{align*}
    where $C_0,C>0$ are universal constants. What is more, for all $R(X)^{1/2}\sigma(X)^{1/2}\leq R\leq\sigma(X)$, setting 
    \[ p_R = \mathbb P\left(\max_{1\leq s\leq T} \|Z_s\|_\infty > R \right), \]
    we have that, for all $t\geq\log m$, with probability at least $1-me^{-t}-p_R$,
    \begin{align*} 
		\|X\|_\infty & \leq \|\E(XX^*)\|_\infty^{1/2}+\|\E(X^*X)\|_\infty^{1/2} \\
        & \qquad + C'\left( (\log m)^{3/4} \sigma(X)^{1/2}\upsilon(X)^{1/2} + \sigma(X)^{2/3}R^{1/3}t + \upsilon(X)t^{1/2} \right),
    \end{align*} 
    where $C'>0$ is a universal constant.
\end{theorem}

As explained in \cite[Section 2.1.4]{vanHandel2}, the condition $\sigma(X)\geq C_0(\log m)^3R(X)$ is satisfied for many natural random matrix models (at least for the appropriate scaling of $T$ with respect to $m$). This will in particular be the case for the ones we will apply Theorem \ref{t: van handel} to. This assumption simply allows to reduce the number of terms involving $\sigma(X)$ and $R(X)$ in the upper bounds on $\E\|X\|_\infty$ and $\|X\|_\infty$, by keeping only the dominating ones, and hence making the presentation lighter. As a consequence of this condition, we also have that the bound from below on the parameter $R$ in the second part of Theorem \ref{t: van handel}, which is $\max(R(X)^{1/2}\sigma(X)^{1/2},\sqrt{2}R(X))$ in \cite[Theorem 2.8]{vanHandel2}, actually simplifies to $R(X)^{1/2}\sigma(X)^{1/2}$ here. While the bound from above $\sigma(X)$ on this parameter is added only to reduce further the number of terms involving $\sigma(X)$ and $R(X)$ in the upper bound on $\|X\|_\infty$.

Later on, we will make use of the following simple observation: the parameter $R(X)$ appearing in Theorem \ref{t: van handel} can be itself upper bounded, for any $p\in\mathbb N$, by the parameter $R_p(X)$ defined as
\begin{equation} \label{eq:def-R_p} R_p(X):=\left(\sum_{s=1}^T \E\,\tr\left(|Z_s|^{2p}\right)\right)^{1/2p} . \end{equation}
Indeed, by ordering of Schatten norms, we have that $\|Z_s\|_\infty \leq \|Z_s\|_{2p} = (\tr(|Z_s|^{2p}))^{1/2p}$ for each $1\leq s\leq T$. And hence,
\begin{align*} 
\E \left(\max_{1\leq s\leq T}\|Z_s\|_\infty^2\right) & \leq \E \left(\max_{1\leq s\leq T} \left(\tr\left(|Z_s|^{2p}\right) \right)^{1/p}\right) \\
& \leq \left(\E \left(\max_{1\leq s\leq T} \tr\left(|Z_s|^{2p}\right) \right) \right)^{1/p} \\
& \leq \left(\sum_{s=1}^T \E \left(\tr\left(|Z_s|^{2p}\right) \right) \right)^{1/p} ,
\end{align*}
where the second inequality is by Jensen inequality. So we do have that, for any $p\in\mathbb N$,
$$ R(X) = \left(\E \left(\max_{1\leq s\leq T}\|Z_s\|_\infty^2\right)\right)^{1/2} \leq \left(\sum_{s=1}^T \E \left(\tr\left(|Z_s|^{2p}\right) \right) \right)^{1/2p} = R_p(X). $$
The parameter $R_p(X)$ has the advantage of being easier to compute than the parameter $R(X)$. And as we will see, for the random matrix $X$ we are interested in, we will get an upper bound on $R(X)$ in terms of $R_p(X)$ that is close to tight by choosing $p$ of order $\log n$.

Finally, it is worth mentioning that Theorem~\ref{t: van handel} is only one particular formulation of several other corollaries which could be stated as a combination of the works \cite{vanHandel} and \cite{vanHandel2}. Indeed, the former derives sharp estimates on the norm of a Gaussian matrix while the latter establishes a universality principle allowing to transfer bounds on the Gaussian model to a more general setting.

\subsection{Proof of Theorem~\ref{th:main-result}} \hfill\vspace{0.1cm}

 We will apply Theorem \ref{t: van handel} above to a random matrix model of the form given by equation \eqref{eq:def-Y}, or more precisely its re-centered version. Concretely, given $W_s$, $1\leq s\leq d$, independent copies of $W$, we define 
 \begin{equation} \label{eq:def-X} X=\frac{1}{d}\sum_{s=1}^d \left(W_s\otimes\overbar{W}_s-\E\left((W_s\otimes\overbar{W}_s\right)\right) , \end{equation}
 which is a sum of $d$ independent and centered $n^2\times n^2$ random matrices. It thus satisfies the hypotheses of Theorem \ref{t: van handel}. Therefore, in order to upper bound its operator norm, our task will reduce to estimating the parameters appearing on the right-hand side of the two inequalities, namely $\sigma(X)$, $\upsilon(X)$ and $R(X)$ (in fact $R_p(X)$ for $p$ large enough, as explained just after Theorem \ref{t: van handel}). This is the content of Lemmas \ref{lem:sigma}, \ref{lem:upsilon} and \ref{lem:R_p} below. 

 \begin{lemma} \label{lem:sigma}
 Let $\W$ be an $n\times n$ random matrix satisfying Assumptions \ref{a:indep entries}, \ref{a: profil variance}, \ref{a: 4 moment} and \ref{a: max variance}. Let $X=\frac{1}{d}\sum_{s=1}^d (\W_s\otimes\overbar\W_s-\E(\W_s\otimes\overbar\W_s))$, where $\W_s$, $1\leq s\leq d$, are independent copies of $\W$. Then, 
	$$ \left\|\E(XX^*)\right\|_\infty, \left\|\E(X^*X)\right\|_\infty \leq \frac{1}{d}\left(1+\frac{C_\beta}{(\log n)^{\alpha_\beta}}\right), $$
 where $C_\beta>0$ is a constant depending only on $\beta$.  
\end{lemma}

\begin{proof}
For every $1\leq s\leq d$, set $Z_s=\W_s\otimes\overbar\W_s-\E(\W_s\otimes\overbar\W_s)$. We then have
\begin{align*} 
\E(XX^*) & = \frac{1}{d^2}\sum_{s,t=1}^d \E\left(Z_sZ_t^*\right) \\
& = \frac{1}{d^2}\sum_{s=1}^d \E\left(Z_sZ_s^*\right) \\
& = \frac{1}{d} \E\left(\left(\W\otimes\overbar\W\right)\left(W\otimes\overbar W\right)^*\right)-\E\left(\W\otimes\overbar\W\right)\E\left(\W\otimes\overbar\W\right)^*
\end{align*}
where the next-to-last equality is because the $Z_s$, $1\leq s\leq d$, are independent and the last equality is because the $Z_s$, $1\leq s\leq d$, are identically distributed.

Next, we have on the one hand
$$ \E\left(W\otimes\overbar W\right) = \sum_{i,j,k,l=1}^n \E\left(W_{ij}\overbar W_{kl}\right) E_{ij}\otimes E_{kl} = \sum_{i,j=1}^n \eta_{ij} E_{ij}\otimes E_{ij}. $$
This is because
$$ \E\left(W_{ij}\overbar W_{kl}\right) = \begin{cases} \eta_{ij} & \text{if} \quad k=i,\ l=j\\ 0 & \text{otherwise} \end{cases}. $$
Therefore,
$$ \E\left(\W\otimes\overbar\W\right)\E\left(\W\otimes\overbar\W\right)^* = \sum_{i,j,k=1}^n \eta_{ij}\eta_{kj} E_{ik}\otimes E_{ik}. $$
On the other hand, we have
\begin{align*} 
\E\left(\left(\W\otimes\overbar\W\right)\left(W\otimes\overbar W\right)^*\right) & = \sum_{i,j,k,i',j',k'=1}^n \E\left(W_{ij}\overbar W_{kj}\overbar W_{i'j'}W_{k'j'}\right) E_{ik}\otimes E_{i'k'} \\
& = \sum_{i,j,k=1}^n \eta_{ij}\eta_{kj} E_{ik}\otimes E_{ik} + R + S + T, 
\end{align*}
where, setting $\zeta_{ij}=\E(W_{ij}^2)$ for each $1\leq i,j\leq n$, we have defined
\begin{align*} 
R & = \sum_{i,j,i',j'=1}^n \eta_{ij}\eta_{i'j'} E_{ii}\otimes E_{i'i'}, \\
S & = \sum_{i,j,k=1}^n \zeta_{ij}\overbar\zeta_{kj} E_{ik}\otimes E_{ki}, \\
T & = \sum_{i,j=1}^n \left( \E|W_{ij}|^4 - 2\eta_{ij}^2 - |\zeta_{ij}|^2\right) E_{ii}\otimes E_{ii}.
\end{align*}
This is because
$$ \E\left(W_{ij}\overbar W_{kj}\overbar W_{i'j'}W_{k'j'}\right) = \begin{cases} \eta_{ij}\eta_{kj} & \text{if} \quad i'=i,\ k'=k,\ j'=j,\ k\neq i \\ \eta_{ij}\eta_{i'j'} & \text{if} \quad k=i,\ k'=i',\ (i'\neq i\ \text{or}\ j'\neq j) \\ \zeta_{ij}\overbar\zeta_{kj} & \text{if} \quad k'=i,\ i'=k,\ j'=j,\ k\neq i \\ 
\E|W_{ij}|^4 & \text{if} \quad k'=i'=k=i,\ j'=j \\ 0 & \text{otherwise} \end{cases}. $$
We thus see that 
$$ \E\left(\left(\W\otimes\overbar\W\right)\left(W\otimes\overbar W\right)^*\right)-\E\left(\W\otimes\overbar\W\right)\E\left(\W\otimes\overbar\W\right)^* = R+S+T,$$
so in order to upper bound $\|\E(XX^*)\|_\infty$, it is enough to upper bound $\|R\|_\infty,\|S\|_\infty,\|T\|_\infty$ and apply the triangle inequality. 

The matrices $R$ and $T$ are diagonal. So we simply have
$$ \|R\|_\infty = \max_{1\leq i,i'\leq n} \left(\sum_{j,j'=1}^n\eta_{ij}\eta_{i'j'}\right) \leq 1, $$
where the inequality is by Assumption \ref{a: profil variance}. And similarly
$$ \|T\|_\infty = \max_{1\leq i\leq n} \left|\sum_{j=1}^n \left( \E|W_{ij}|^4 - 2\eta_{ij}^2 - |\zeta_{ij}|^2\right)\right| \leq \max_{1\leq i\leq n} \left(\sum_{j=1}^n \left(\E|W_{ij}|^4+2\eta_{ij}^2+|\zeta_{ij}|^2\right) \right).$$
By Assumption \ref{a: 4 moment}, for each $1\leq i,j\leq n$, $\E|W_{ij}|^4\leq\mathbf C^2 4^\beta\eta_{ij}^2$, and by H\"older inequality,
$$ \sum_{j=1}^n \eta_{ij}^2 \leq \left(\max_{1\leq j\leq n} \eta_{ij}\right) \left(\sum_{j=1}^n\eta_{ij} \right) \leq \frac{\mathbf C'}{(\log n)^{\alpha_\beta}},$$
where the last inequality is by Assumptions \ref{a: profil variance} and \ref{a: max variance}. We thus get, setting $\mathbf C_\beta=\mathbf C^2\mathbf C' 4^\beta$, 
$$ \|T\|_\infty \leq \frac{\mathbf C_\beta +3\mathbf C'}{(\log n)^{\alpha_\beta}}. $$
For the matrix $S$, we just have to observe that, denoting by $F$ the so-called flip matrix, i.e.~$F=\sum_{i,k=1}^n E_{ik}\otimes E_{ki}$, we have that 
\[ FS = \sum_{i,j,k=1}^n \zeta_{ij}\overbar\zeta_{kj} E_{ii}\otimes E_{kk}, \]
is diagonal as well. Now, since $F$ is a unitary matrix, it does not change the operator norm, so that
$$ \|S\|_\infty = \|FS\|_\infty = \max_{1\leq i,k\leq n} \left|\sum_{j=1}^n \zeta_{ij}\overbar\zeta_{kj} \right|. $$
Next, by H\"older inequality, for each $1\leq i,k\leq n$,
$$ \left|\sum_{j=1}^n\zeta_{ij}\overbar\zeta_{kj}\right| \leq \left(\max_{1\leq j\leq n} |\zeta_{ij}|\right) \left(\sum_{j=1}^n|\zeta_{kj}|\right) \leq \frac{\mathbf C'}{(\log n)^{\alpha_\beta}}, $$
where the last inequality is by Assumptions \ref{a: profil variance} and \ref{a: max variance}, after observing that, for each $1\leq i,j\leq n$, $|\zeta_{ij}|\leq \eta_{ij}$. And we thus have
$$ \|S\|_\infty \leq \frac{\mathbf C'}{(\log n)^{\alpha_\beta}}. $$


Putting everything together, we obtain exactly the announced result, namely 
$$ \|\E(XX^*)\|_\infty \leq \frac{1}{d}\left(1+\frac{4\mathbf C'+\mathbf C_\beta}{(\log n)^{\alpha_\beta}} \right). $$
The reasoning is entirely similar for $\|\E(X^*X)\|_\infty$.
\end{proof}

\begin{lemma} \label{lem:upsilon}
 Let $\W$ be an $n\times n$ random matrix satisfying Assumptions \ref{a:indep entries}, \ref{a: profil variance}, \ref{a: 4 moment} and \ref{a: max variance}. Let $X=\frac{1}{d}\sum_{s=1}^d (\W_s\otimes\overbar\W_s-\E(\W_s\otimes\overbar\W_s))$, where $\W_s$, $1\leq s\leq d$, are independent copies of $\W$. Then, 
	$$ \left\|\Cov(X)\right\|_\infty \leq \frac{1}{d}\times\frac{C_\beta}{(\log n)^{2\alpha_\beta}}, $$
 where $C_\beta>0$ is a constant depending only on $\beta$. 
\end{lemma}

\begin{proof}
 For every $1\leq s\leq d$, set $Z_s=\W_s\otimes\overbar\W_s-\E(\W_s\otimes\overbar\W_s)$. Since the $Z_s$, $1\leq s\leq d$, are independent and identically distributed, we have
 $$ \Cov(X) = \frac{1}{d^2}\sum_{s=1}^d \Cov(Z_s) = \frac{1}{d} \Cov\left(\W\otimes\overbar\W-\E(\W\otimes\overbar\W)\right). $$
 It is not difficult to check that we can rewrite
 \begin{align*}
     & \left\|\Cov\left(\W\otimes\overbar\W-\E(\W\otimes\overbar\W)\right)\right\|_\infty \\
     & \quad = \sup_{\|M\|_2\leq 1} \E\left(\left|\tr\left[M\left(\W\otimes\overbar\W-\E(\W\otimes\overbar\W)\right)\right]\right|^2\right) \\
     & \quad = \sup_{\|M\|_2\leq 1} \tr\left[\left(M\otimes\overbar M\right)\left(\E\left(\W\otimes\overbar\W\otimes\overbar W\otimes W\right)-\E(\W\otimes\overbar\W)\otimes\E(\overbar\W\otimes\W)\right)\right].
 \end{align*} 

Next, we have already seen that
$$ \E\left(W\otimes\overbar W\right) = \sum_{i,j=1}^n \eta_{ij} E_{ij}\otimes E_{ij}, $$
so that
$$ \E\left(\W\otimes\overbar\W\right)\otimes\E\left(\overbar\W\otimes\W\right) = \sum_{i,j,k,l=1}^n \eta_{ij}\eta_{kl} E_{ij}\otimes E_{ij}\otimes E_{kl}\otimes E_{kl}. $$
On the other hand, we have
\begin{align*} 
\E\left(\W\otimes\overbar\W\otimes\overbar\W\otimes\W\right) & = \sum_{i,j,k,l,i',j',k',l'=1}^n \E\left(W_{ij}\overbar W_{kl}\overbar W_{i'j'}W_{k'l'}\right) E_{ij}\otimes E_{kl}\otimes E_{i'j'}\otimes E_{k'l'} \\
& = \sum_{i,j,i',j'=1}^n \eta_{ij}\eta_{i'j'} E_{ij}\otimes E_{ij}\otimes E_{i'j'}\otimes E_{i'j'} + R + S + T, 
\end{align*}
where, setting $\zeta_{ij}=\E(W_{ij}^2)$ for each $1\leq i,j\leq n$, we have defined
\begin{align*} 
R & = \sum_{i,j,k,l=1}^n \eta_{ij}\eta_{kl} E_{ij}\otimes E_{kl}\otimes E_{ij}\otimes E_{kl}, \\
S & = \sum_{i,j,k,l=1}^n \zeta_{ij}\overbar\zeta_{kl} E_{ij}\otimes E_{kl}\otimes E_{kl}\otimes E_{ij}, \\
T & = \sum_{i,j=1}^n \left( \E|W_{ij}|^4 - 2\eta_{ii}\eta_{jj} - \zeta_{ii}\overbar\zeta_{jj}\right) E_{ij}\otimes E_{ij}\otimes E_{ij}\otimes E_{ij}.
\end{align*}
This is because
$$ \E\left(W_{ij}\overbar W_{kj}\overbar W_{i'j'}W_{k'l'}\right) = \begin{cases} \eta_{ij}\eta_{i'j'} & \text{if} \quad k=i,\ l=j,\ k'=i',\ l'=j',\ (i'\neq i\ \text{or}\ j'\neq j) \\ \eta_{ij}\eta_{kl} & \text{if} \quad i'=i,\ j'=j,\ k'=k,\ l'=l,\ (k\neq i\ \text{or}\ l\neq j) \\ \zeta_{ij}\overbar\zeta_{kl} & \text{if} \quad k'=i,\ l'=j,\ i'=k,\ j'=l,\ (k\neq i\ \text{or}\ l\neq j) \\ 
\E|W_{ij}|^4 & \text{if} \quad k'=i'=k=i,\ l'=j'=l=j \\ 0 & \text{otherwise} \end{cases}. $$
We thus see that 
$$ \E\left(\W\otimes\overbar\W\otimes\overbar W\otimes W\right)-\E(\W\otimes\overbar\W)\otimes\E(\overbar\W\otimes\W) = R+S+T,$$
so in order to upper bound $\|\Cov(X)\|_\infty$, it is enough to upper bound $\sup_{\|M\|_2\leq 1}|\tr[(M\otimes\overbar M)U]|$ for $U=R,S,T$ and apply the triangle inequality. 

Now, given $M\in\mathcal M_{n^2}(\mathbb C)$ such that $\|M\|_2\leq 1$, we first have by H\"older inequality
\begin{align*}
    \left|\tr\left[(M\otimes\overbar M)R\right]\right| & = \sum_{i,j,k,l=1}^n |M_{ikjl}|^2\eta_{ij}\eta_{kl} \\
    & \leq \left(\sum_{i,j,k,l=1}^n |M_{ikjl}|^2\right)\left(\max_{1\leq i,j,k,l\leq n}\eta_{ij}\eta_{kl}\right) \\
    & \leq \frac{\mathbf C'^2}{(\log n)^{2\alpha_\beta}},
\end{align*} 
where the last inequality is by Assumption \ref{a: profil variance}, recalling that $\sum_{i,j,k,l=1}^n |M_{ikjl}|^2=\|M\|_2^2\leq 1$. In complete analogy, we have
\begin{align*}
    \left|\tr\left[(M\otimes\overbar M)S\right]\right| & = \left|\sum_{i,j,k,l=1}^n M_{ikjl}\overbar M_{kilj}\zeta_{ij}\overbar\zeta_{kl}\right| \\
    & \leq \left(\sum_{i,j,k,l=1}^n |M_{ikjl}|^2\right)\left(\max_{1\leq i,j,k,l\leq n}|\zeta_{ij}\overbar\zeta_{kl}|\right) \\
    & \leq \frac{\mathbf C'^2}{(\log n)^{2\alpha_\beta}},
\end{align*} 
where the last inequality is by Assumption \ref{a: profil variance}, after observing that, for each $1\leq i,j\leq n$, $|\zeta_{ij}|\leq\eta_{ij}$. And finally, we have
\begin{align*}
    \left|\tr\left[(M\otimes\overbar M)T\right]\right| & = \left|\sum_{i,j=1}^n |M_{iijj}|^2\left( \E|W_{ij}|^4 - 2\eta_{ii}\eta_{jj} - \zeta_{ii}\overbar\zeta_{jj}\right)\right| \\
    & \leq \sum_{i,j=1}^n |M_{iijj}|^2\left( \E|W_{ij}|^4 + 2\eta_{ii}\eta_{jj} + |\zeta_{ii}\overbar\zeta_{jj}|\right) \\
    & \leq \left(\sum_{i,j=1}^n |M_{iijj}|^2\right)\left(\max_{1\leq i,j\leq n}\left( \E|W_{ij}|^4 + 2\eta_{ii}\eta_{jj} + |\zeta_{ii}\overbar\zeta_{jj}|\right)\right) .
\end{align*} 
By Assumption \ref{a: 4 moment}, for each $1\leq i,j\leq n$, $\E|W_{ij}|^4\leq\mathbf C^2 4^\beta\eta_{ij}^2$. We thus get by Assumption \ref{a: max variance}, setting $\mathbf C_\beta=\mathbf C^2\mathbf C'^2 4^\beta$, 
$$ \left|\tr\left[(M\otimes\overbar M)T\right]\right| \leq \frac{\mathbf C_\beta+3\mathbf C'^2}{(\log n)^{2\alpha_\beta}}. $$

Putting everything together, we eventually obtain  
$$ \|\Cov(X)\|_\infty \leq \frac{1}{d}\times\frac{5\mathbf C'^2+\mathbf C_\beta}{(\log n)^{2\alpha_\beta}}, $$
which is exactly the announced result.
\end{proof}

\begin{lemma} \label{lem:R_p}
 Let $\W$ be an $n\times n$ random matrix satisfying Assumptions \ref{a:indep entries}, \ref{a: profil variance}, \ref{a: 4 moment} and \ref{a: max variance}. Let $X=\frac{1}{d}\sum_{s=1}^d (\W_s\otimes\overbar\W_s-\E(\W_s\otimes\overbar\W_s))$, where $\W_s$, $1\leq s\leq d$, are independent copies of $\W$. Then, for $c\log n\leq p\leq C\log n$,
	$$ R_p(X)\leq \frac{C'}{d}, $$
 where $c,C,C'>0$ are universal constants. 
\end{lemma}

\begin{proof}
For every $1\leq s\leq d$, set $Z_s=\W_s\otimes\overbar\W_s-\E(\W_s\otimes\overbar\W_s)$. Since the $Z_s$, $1\leq s\leq d$, are identically distributed, according to $Z=\W\otimes\overbar\W-\E(\W\otimes\overbar\W)$, we have
$$ R_p(X) = \frac{1}{d}\left(\sum_{s=1}^d\E\tr|Z_s|^p\right)^{1/p} = \frac{d^{1/p}}{d}\left(\E\tr|Z|^p\right)^{1/p} \leq \frac{C'}{d} \left(\E\tr|Z|^p\right)^{1/p}, $$
where the last inequality is valid for $p\geq c\log n$ (recalling that $d\leq n^2$). And by H\"older inequality, we have
$$ \left(\E\tr|Z|^p\right)^{1/p} = \left(\E\|Z\|_p^p\right)^{1/p} \leq n^{1/p}\left(\E\|Z\|_\infty^p\right)^{1/p} \leq C'\left(\E\|Z\|_\infty^p\right)^{1/p} , $$
where again the last inequality is valid for $p\geq c\log n$. Now, by the triangle inequality (twice), we have
\begin{align*}
    \left(\E\|Z\|_\infty^p\right)^{1/p} & \leq \left(\E\left(\left\|\W\otimes\overbar\W\right\|_\infty+\left\|\E\left(\W\otimes\overbar\W\right)\right\|_\infty\right)^p\right)^{1/p} \\
    & \leq \left(\E\left\|\W\otimes\overbar\W\right\|_\infty^p\right)^{1/p} + \left\|\E\left(\W\otimes\overbar\W\right)\right\|_\infty \\
    & = \left(\E\left\|\W\right\|_\infty^{2p}\right)^{1/p} + \left\|\E\left(\W\otimes\overbar\W\right)\right\|_\infty.
\end{align*} 
We have already seen that $\E(\W\otimes\overbar\W)=\sum_{i,j=1}^n\eta_{ij}E_{ij}\otimes E_{ij}$. This implies that the matrices $\E(\W\otimes\overbar\W)$ and $\eta$ have the same (non-zero) singular values. In particular $\|\E(\W\otimes\overbar\W)\|_\infty=\|\eta\|_\infty$. Now, by Schur's inequality we know that
$$ \|\eta\|_\infty \leq \left(\max_{1\leq i\leq n} \left(\sum_{k=1}^n\eta_{ik}\right) \max_{1\leq j\leq n}\left(\sum_{l=1}^n\eta_{lj}\right)\right)^{1/2} \leq 1, $$
where the last inequality is by Assumption \ref{a: profil variance}. And we thus have $\|\E(\W\otimes\overbar\W)\|_\infty\leq 1$.
Moreover, we have by \cite[Corollary~3.5]{MR3531673} or \cite[Theorem~4.4]{MR3878726} that, for $p\leq C\log n$,
\begin{align*} 
\left(\E\|W\|_\infty^{2p}\right)^{1/p} & \leq C'\left(\max_{1\leq i\leq n}\left(\sum_{j=1}^n\eta_{ij}\right)^{1/2} + (\log n)^{\max(\beta,1)/2}\max_{1\leq i,j\leq n}\eta_{ij}^{1/2}\right)^2 \\
& \leq C'\left(1 + \frac{\mathbf C'}{(\log n)^{\alpha_\beta/2-\max(\beta,1)/2}}\right)^2 \\
& \leq C''.
\end{align*}
The next-to-last inequality is by Assumptions \ref{a: profil variance} and \ref{a: max variance}. And the last inequality is because $\alpha_\beta\geq\max(\beta,1)$.

Putting everything together, we finish the proof. 
\end{proof}

With Lemmas \ref{lem:sigma}, \ref{lem:upsilon} and \ref{lem:R_p} at hand, we are now ready to prove the upper bound on $\|Y-\mathbb E(Y)\|_\infty$, for the random matrix $Y$ defined in equation \eqref{eq:def-Y}, promised in Theorem~\ref{th:main-result}. 

\begin{proof}[Proof of Theorem~\ref{th:main-result}]
We have estimated in Lemmas \ref{lem:sigma}, \ref{lem:upsilon} and \ref{lem:R_p} the parameters $\sigma(X)$, $\upsilon(X)$ and $R_p(X)$ for the random matrix $X=Y-\E(Y)$ defined in equation \eqref{eq:def-X}. Note that, by the triangle inequality, the proof of Lemma \ref{lem:sigma} gives not only that $\sigma(X)\leq (1+C_\beta/(\log n)^{\alpha_\beta})/\sqrt{d}$ but also that $\sigma(X)\geq (1-C_\beta/(\log n)^{\alpha_\beta})/\sqrt{d}$. We additionally have by Lemma \ref{lem:R_p} that $R(X)\leq C/d$. Hence, as soon as $d\geq(\log n)^6$, we do have that $\sigma(X)\geq C'_\beta(\log n)^3R(X)$. So the assumptions of the first estimate in Theorem \ref{t: van handel} are indeed satisfied, and we get applying it
\begin{align*} 
\E\left\|Y-\E(Y)\right\|_\infty & \leq \frac{2}{\sqrt{d}}\left(1 + \hat{C}_\beta\left(\frac{1}{(\log n)^{\alpha_\beta}} + \frac{(\log n)^{3/4}}{(\log n)^{\alpha_\beta/2}} + \frac{\log n}{d^{1/12}} \right)\right) \\ 
& \leq \frac{2}{\sqrt{d}}\left(1 + \hat{C}_\beta\left(\frac{1}{(\log n)^2} + \frac{1}{(\log n)^{1/4}} + \frac{1}{(\log n)^{\delta/12}} \right)\right) \\ 
& \leq \frac{2}{\sqrt{d}}\left(1+\frac{3\hat{C}_\beta}{(\log n)^{\min(\delta/12,1/4)}} \right) , 
\end{align*}
where the second inequality is by using the hypotheses that $\alpha_\beta\geq 2$ and $d\geq(\log n)^{12+\delta}$.

We can further apply the second estimate in Theorem \ref{t: van handel}. We choose $R=C/d^{3/4}$, so that the assumption $R(X)^{1/2}\sigma(X)^{1/2}\leq R\leq\sigma(X)$ is indeed satisfied. With this choice, we have by Markov inequality that
\[ p_R = \mathbb P\left(\max_{1\leq s\leq d} \|Z_s\|_\infty > R \right) \leq \frac{1}{R}\, \E\left(\underset{1\leq s\leq d}{\max} \|Z_s\|_\infty\right) \leq \frac{R(X)}{R} \leq \frac{1}{d^{1/4}}, \]
where the next to last inequality is by Jensen inequality. Taking $t=3\log n$, we then obtain that, with probability at least $1-1/n-1/d^{1/4}$, hence a fortiori at least $1-2/d^{1/4}$ (recalling that $d\leq n^2\leq n^4$),
\begin{align*} 
\left\|Y-\E(Y)\right\|_\infty & \leq \frac{2}{\sqrt{d}}\left(1 + \hat{C}_\beta'\left(\frac{1}{(\log n)^{\alpha_\beta}} + \frac{(\log n)^{3/4}}{(\log n)^{\alpha_\beta/2}} + \frac{\log n}{d^{1/12}} + \frac{(\log n)^{1/2}}{(\log n)^{\alpha_\beta}} \right)\right) \\ 
& \leq \frac{2}{\sqrt{d}}\left(1 + \hat{C}_\beta'\left(\frac{1}{(\log n)^2} + \frac{1}{(\log n)^{1/4}} + \frac{1}{(\log n)^{\delta/12}} + \frac{1}{(\log n)^{3/2}} \right)\right) \\ 
& \leq \frac{2}{\sqrt{d}}\left(1+\frac{4\hat{C}_\beta'}{(\log n)^{\min(\delta/12,1/4)}} \right) , 
\end{align*}
where the second inequality is by using the hypotheses that $\alpha_\beta\geq 2$ and $d\geq(\log n)^{12+\delta}$.
\end{proof}

\section{Implications in terms of random quantum channels} \label{sec:application}

The goal of this section is to use the previous estimates in order to design an optimal quantum expander. Let $W$ be an $n\times n$ random matrix satisfying Assumptions \ref{a:indep entries}, \ref{a: profil variance}, \ref{a: 4 moment} and \ref{a: max variance} and let $Y$ be defined as in equation \eqref{eq:def-Y}. The result of the previous section ensures that $Y$ typically concentrates around its expectation. In order for this to imply that $Y$ typically has a large gap between its two largest singular values (and thus corresponds to an optimal quantum expander), we need to ensure that $\E(Y)$ itself has a large gap between its two largest singular values. We will now make the following extra assumption on the variance matrix $(\eta_{ij})_{1\leq i,j\leq n}$ to ensure this, and to additionally guarantee that the CP map corresponding to $Y$ is on average TP. 

\begin{assumption} \label{a:eta}
The $n\times n$ matrix $\eta:=(\eta_{ij})_{1\leq i,j\leq n}$ is such that $\eta^t$ is the transition matrix of an irreducible Markov chain on $\{1,\ldots, n\}$ satisfying: there exists a universal constant $\mathbf C''>0$ such that 
 $$ s_2(\eta)\leq \frac{\mathbf C''}{\sqrt{d}} , $$ 
 where $s_2(\eta)$ denotes the second largest singular value of $\eta$.  
\end{assumption}

Note that the above assumption ensures that, in addition to being doubly sub-stochastic, $\eta$ is also left-stochastic, i.e.~for each $1\leq j\leq n$, $\sum_{i=1}^n \eta_{ij}=1$. This implies that $s_1(\eta)=\lambda_1(\eta)=1$. The assumption on $s_2(\eta)$ provides, among other, a quantitative control on the rate of convergence of the corresponding Markov chain to its equilibrium measure. The simplest example is when $\eta=J/n$, where $J$ is the $n\times n$ matrix whose entries are all equal to $1$. In this case, the leading eigenvector is the uniform probability $u$ on $\{1,\ldots,n\}$ and $s_2(\eta)=0$, so there is convergence to $u$ in only $1$ step. More generally, we allow $\eta$ to have a unique fixed probability which is not necessarily uniform and a second largest singular value which is small but not necessarily $0$.

\begin{example} 
Let us look at the three examples of random matrix models presented in Section \ref{sec:model}, which were shown to already satisfy Assumptions \ref{a:indep entries}, \ref{a: profil variance}, \ref{a: 4 moment} and \ref{a: max variance}, and see what extra restriction is needed in order to satisfy Assumption \ref{a:eta} as well. First note that, in all examples, the considered random matrix $W$ has an associated variance matrix $\eta$ which is doubly stochastic. In Example \ref{ex:bounded-moments}, we have $\eta=J/n$, so it satisfies Assumption \ref{a:eta} since $s_2(\eta)=0$. For Example \ref{ex:gaussian}, we will restrict the variance profile of the Gaussian matrix to those satisfying Assumption \ref{a:eta}. Finally, for Example \ref{ex:sparse}, we will choose the degree $r$ of the regular graph equal to the Kraus rank $d$, and we will impose that the graph is itself an optimal expander. This can for instance be achieved by picking such $d$-regular graph uniformly at random (see \cite{Sarid} and references therein).
\end{example}

With Assumption \ref{a:eta} at hand, we can now state and prove the main results of this note. We recall that, given a CP map $\Phi$, we denote by $|\lambda_1(\Phi)|,|\lambda_2(\Phi)|$ its largest and second largest (in modulus) eigenvalues, and by $s_1(\Phi),s_2(\Phi)$ its largest and second largest singular values. 

\begin{theorem} \label{th:expander}
 Let $n, d\in \mathbb{N}$ with $d\geq(\log n)^{12}$ and consider $W_1,\ldots,W_d$ independent copies of an $n\times n$ random matrix $W$ satisfying Assumptions \ref{a:indep entries}, \ref{a: profil variance}, \ref{a: 4 moment}, \ref{a: max variance} and \ref{a:eta}. Define the random CP map $\Phi$ on $\mathcal M_n(\mathbb C)$, having Kraus rank at most $d$, as
$$\Phi:X\in\mathcal M_n(\mathbb C)\mapsto \frac{1}{d}\sum_{s=1}^d W_sXW_s^*\in\mathcal M_n(\mathbb C). $$
Then, $\Phi$ satisfies on average the TP condition \eqref{eq:TP-condition} and is such that
 $$ \mathbb P\left( s_1(\Phi) \geq 1 - \frac{C_\beta}{\sqrt{d}} \ \text{and} \ s_2(\Phi) \leq \frac{C_\beta}{\sqrt{d}} \right) \geq 1-\frac{2}{d^{1/4}}, $$
 where $C_\beta>0$ is a constant depending only on $\beta$.
\end{theorem}	

\begin{proof}
We start by noting that 
$$\E\left(W^*W\right) = \sum_{i,j,k=1}^n \E\left(\overbar{W}_{ji}W_{jk}\right) E_{ik} = \sum_{i=1}^n \left(\sum_{j=1}^n \eta_{ji} \right) E_{ii}= \sum_{i=1}^n E_{ii}= I, $$
where the second equality is because $\E(\overbar{W}_{ji}W_{jk})=0$ if $i\neq k$ and $\E(\overbar{W}_{ji}W_{jk})=\eta_{ji}$ if $i=k$, while the third equality is because $\sum_{j=1}^n \eta_{ji}=1$. 
This shows that $\Phi$ satisfies on average the TP condition \eqref{eq:TP-condition}, since 
\[ \mathbb E\left( \frac{1}{d}\sum_{s=1}^d W_s^*W_s \right) = \frac{1}{d}\sum_{s=1}^d \mathbb E\left(W_s^*W_s\right) = I.\]

Similarly, we have
$$\E\left(W\otimes\overbar{W}\right) = \sum_{i,j,k,l=1}^n \E\left(W_{ij}\overbar{W}_{kl}\right) E_{ij}\otimes E_{kl} = \sum_{i,j=1}^n\eta_{ij} E_{ij}\otimes E_{ij},$$
where the second equality is because $\E(W_{ij}\overbar{W}_{kl})=0$ if $i\neq k$ or $j\neq l$ and $\E(W_{ij}\overbar{W}_{kl})=\eta_{ij}$ if $i=k$ and $j=l$. And thus,
$$
\E (M_{\Phi})= \mathbb E\left( \frac{1}{d}\sum_{s=1}^d W_s\otimes\overbar W_s \right) = \frac{1}{d}\sum_{s=1}^d \mathbb E\left(W_s\otimes\overbar W_s\right) = \sum_{i,j=1}^n \eta_{ij}E_{ij}\otimes E_{ij}.
$$
So clearly, the (non-zero) singular values of $\E (M_{\Phi})$ coincide with those of $\eta$. 
As a consequence, we have $s_1(\E(M_{\Phi}))=s_1(\eta)=1$ and by Assumption \ref{a:eta} $s_2(\E(M_{\Phi}))=s_2(\eta)\leq \mathbf C''/\sqrt{d}$. 

Next, by Weyl inequalities for singular values (which are a consequence of the minimax principle), we get
\begin{align*}
    s_1\left(M_\Phi\right) & \geq s_1\left(\E(M_{\Phi})\right) - s_1\left(M_\Phi-\E(M_{\Phi})\right) = 1 - s_1\left(M_\Phi-\E(M_{\Phi})\right), \\
    s_2\left(M_\Phi\right) & \leq s_2\left(\E(M_{\Phi})\right) + s_1\left(M_\Phi-\E(M_{\Phi})\right) \leq \frac{\mathbf C''}{\sqrt{d}} + s_1\left(M_\Phi-\E(M_{\Phi})\right). \\
\end{align*}
Now, we know by Theorem \ref{th:main-result} (applied with $\delta=0$) that, with probability at least $1-2/d^{1/4}$, $s_1(M_\Phi-\E(M_{\Phi}))\leq C_\beta/\sqrt{d}$. And the conclusion of Theorem \ref{th:expander} follows.

\end{proof}

Let us note that, in Theorem \ref{th:expander}, we only guarantee that the random CP map we consider satisfies the TP condition \eqref{eq:TP-condition} on average. It is in fact possible to modify slightly the construction in Theorem \ref{th:expander} to get a random CP map $\tilde{\Phi}$ which is exactly TP, but still has singular values that are, with high probability, close to those of $\Phi$. Indeed, setting
\[ \Sigma= \frac{1}{d}\sum_{s=1}^d W_s^*W_s , \]
we will show that, not only $\E(\Sigma)=I$, but also that $\|\Sigma-I\|_{\infty}$ is small with high probability. One could thus set $\tilde{W}_s=W_s\Sigma^{-1/2}$, $1\leq s\leq d$, and define the random CP map $\tilde{\Phi}$, which is TP by construction, as
\[ \tilde\Phi:X\in\mathcal M_n(\mathbb C)\mapsto \frac{1}{d}\sum_{s=1}^d \tilde{W}_sX\tilde{W}_s^*\in\mathcal M_n(\mathbb C). \]
As we will see, the latter is such that, with high probability, its singular values do not differ by much from those of $\Phi$. This will be a straightforward consequence of the following lemma.

\begin{lemma} \label{lem:Sigma}
Let $\W$ be an $n\times n$ random matrix satisfying Assumptions \ref{a:indep entries}, \ref{a: profil variance}, \ref{a: 4 moment} and \ref{a: max variance}. Let $\Sigma=\frac{1}{d}\sum_{s=1}^d W_s^*W_s$, where $\W_s$, $1\leq s\leq d$, are independent copies of $\W$. If $d\geq(\log n)^{4}$, then 
$$ \E \Vert \Sigma-I \Vert_\infty \leq \frac{C_\beta}{\sqrt{d}},$$
and what is more, 
$$ \mathbb P\left( \Vert \Sigma-I \Vert_\infty \leq \frac{C_\beta'}{\sqrt{d}} \right) \geq 1-\frac{2}{d^{1/4}}, $$
where $C_\beta,C_\beta'>0$ are constants depending only on $\beta$.
\end{lemma}

The proof of Lemma \ref{lem:Sigma} is entirely analogous to that of Theorem \ref{th:main-result}. Namely it consists in applying Theorem \ref{t: van handel} to 
$$\Sigma-\E(\Sigma)=\frac{1}{d}\sum_{s=1}^d \left(W_s^*W_s-\E(W_s^*W_s)\right),$$ 
which is a sum of $d$ independent and centered $n\times n$ (self-adjoint) random matrices. In order to do that, we have to estimate the parameters $\sigma,\upsilon,R_p$ for $\Sigma-\E(\Sigma)$. As we will see, the computations are very similar to those of Lemmas \ref{lem:sigma}, \ref{lem:upsilon} and \ref{lem:R_p}, which estimate these parameters for $Y-\E(Y)$. We will thus skip some of the details when repeating the arguments.

\begin{proof}
Set $Z=W^*W-\E(W^*W)$. First, we have
$$ \E\left((\Sigma-I)^2\right) = \frac{1}{d}\E\left(Z^2\right) = \frac{1}{d}\left(\E\left(W^*WW^*W\right)-\E(W^*W)^2\right) .$$
We have already seen that $\E(W^*W)=I$. On the other hand, we have
$$ \E\left(W^*WW^*W\right) = \sum_{i=1}^n \left(\sum_{j,k=1}^n\eta_{ji}\eta_{ki}\right)E_{ii} + \sum_{i=1}^n \left(\sum_{j=1}^n\E|W_{ji}|^4-\eta_{ji}^2\right)E_{ii} . $$
The first term on the right-hand side of the above equality is equal to $I$, because, for each $1\leq i\leq n$, $\sum_{j=1}^n\eta_{ji}=1$. We thus have 
$$ \E\left(W^*WW^*W\right)-\E(W^*W)^2 = \sum_{i=1}^n \left(\sum_{j=1}^n\E|W_{ji}|^4-\eta_{ji}^2\right)E_{ii}, $$
and therefore
\begin{align*} 
\left\|\E\left(W^*WW^*W\right)-\E(W^*W)^2\right\|_\infty & = \max_{1\leq i\leq n} \left(\sum_{j=1}^n\E|W_{ji}|^4-\eta_{ji}^2\right) \\
& \leq  (\mathbf C^24^\beta-1)\max_{1\leq i\leq n} \left(\sum_{j=1}^n\eta_{ji}^2\right) \\
& \leq \frac{(\mathbf C^24^\beta-1)\mathbf C'}{(\log n)^{\alpha_\beta}}. 
\end{align*}
Hence, we have shown that
$$ \sigma(\Sigma-I) \leq \frac{C_\beta}{(\log n)^{\alpha_\beta/2}}\times\frac{1}{\sqrt{d}} \leq \frac{C_\beta}{\log n}\times\frac{1}{\sqrt{d}}, $$
where the last inequality is because $\alpha_\beta\geq 2$.

Second, we have
$$ \Cov\left(\Sigma-I\right) = \frac{1}{d}\Cov(Z) = \frac{1}{d}\Cov\left(W^*W-\E(W^*W)\right), $$
and we can rewrite
\begin{align*}
    & \left\|\Cov\left(W^*W-\E(W^*W)\right)\right\|_\infty \\
    & \quad = \sup_{\|M\|_2\leq 1} \tr\left[ \left(M\otimes\overbar M\right)\left(\E\left(W^*W\otimes\overbar W^*\overbar W\right) - \E\left(W^*W\right)\otimes\E\left(\overbar W^*\overbar W\right)\right) \right] .
\end{align*} 
We already know that $\E(W^*W)=I$, while we have
\begin{align*}
    & \E\left(W^*W\otimes\overbar W^*\overbar W\right) \\
    & \quad = \sum_{i,j=1}^n\left(\sum_{k,l=1}^n\eta_{ki}\eta_{lj}\right)E_{ii}\otimes E_{jj} + \sum_{i,j=1}^n\left(\sum_{k=1}^n\eta_{ki}\eta_{kj}\right)E_{ij}\otimes E_{ij} \\
    & \qquad + \sum_{i,j=1}^n\left(\sum_{k=1}^n\zeta_{ki}\overbar\zeta_{kj}\right)E_{ij}\otimes E_{ji} + \sum_{i=1}^n\left(\sum_{k=1}^n\left(\E|W_{ki}|^4-2\eta_{ki}^2-|\zeta_{ki}|^2\right)\right)E_{ii}\otimes E_{ii}.
\end{align*}
The first term on the right-hand side of the above equality is equal to $I$, because, for each $1\leq i\leq n$, $\sum_{j=1}^n\eta_{ji}=1$. We thus have 
\begin{align*}
    & \E\left(W^*W\otimes\overbar W^*\overbar W\right) - \E\left(W^*W\right)\otimes\E\left(\overbar W^*\overbar W\right) \\
    & \quad = \sum_{i,j=1}^n\left(\sum_{k=1}^n\eta_{ki}\eta_{kj}\right)E_{ij}\otimes E_{ij} + \sum_{i,j=1}^n\left(\sum_{k=1}^n\zeta_{ki}\overbar\zeta_{kj}\right)E_{ij}\otimes E_{ji} \\
    & \qquad + \sum_{i=1}^n\left(\sum_{k=1}^n\left(\E|W_{ki}|^4-2\eta_{ki}^2-|\zeta_{ki}|^2\right)\right)E_{ii}\otimes E_{ii}.
\end{align*}
Now, it is easy to upper bound, for each term on the right-hand side of the above equality, its trace with $M\otimes\overbar M$, for $M\in\mathcal M_n(\mathbb C)$ such that $\|M\|_2\leq 1$, exactly as it was done in the proof of Lemma \ref{lem:upsilon}. We get that this quantity is upper bounded by $\mathbf C'/(\log n)^{\alpha_\beta}$ for the first two terms and by $(\mathbf C^24^\beta+3)\mathbf C'/(\log n)^{\alpha_\beta}$ for the third term. Hence, we have shown that
$$ \upsilon(\Sigma-I) \leq \frac{C_\beta}{(\log n)^{\alpha_\beta/2}}\times\frac{1}{\sqrt{d}} \leq \frac{C_\beta}{\log n}\times\frac{1}{\sqrt{d}}, $$
where the last inequality is because $\alpha_\beta\geq 2$.

Finally, arguing as in the proof of Lemma \ref{lem:R_p}, we have for $p\geq c\log n$
$$ R_p(\Sigma-I) \leq \frac{(dn)^{1/p}}{d}\left(\E\|Z\|_\infty^p\right)^{1/p} \leq \frac{C'}{d}\left(\E\|Z\|_\infty^p\right)^{1/p} . $$
And next for $p\leq C\log n$, 
$$ \left(\E\|Z\|_\infty^p\right)^{1/p} \leq \left(\E\left\|W^*W\right\|_\infty^p\right)^{1/p} + \left\|\E\left(W^*W\right)\right\|_\infty = \left(\E\left\|\W\right\|_\infty^{2p}\right)^{1/p} + \left\|I\right\|_\infty \leq C''+1, $$
where the last inequality is by using the computation already performed in the proof of Lemma \ref{lem:R_p}. Hence, we have shown that, for $c\log n\leq p\leq C\log n$,
$$ R_p(\Sigma-I) \leq \frac{C'''}{d}.$$

With these three ingredients at hand, we can straightforwardly conclude using Theorem \ref{t: van handel}. The upper bounds $\sigma_*(\Sigma-I)$ and $R_*(\Sigma-I)$ that we use on $\sigma(\Sigma-I)$ and $R(\Sigma-I)$, namely $C_\beta/(\log n)\sqrt{d}$ and $C'''/d$ respectively, are indeed such that $\sigma_*(\Sigma-I)\geq C_\beta'(\log n)^3R_*(\Sigma-I)$ as soon as $d\geq(\log n)^4$. Hence, by applying the first estimate, we obtain
$$ \E \Vert \Sigma-I \Vert_\infty \leq \frac{C_\beta}{\sqrt{d}} \left( \frac{1}{\log n} + \frac{(\log n)^{3/4}}{\log n} + \frac{\log n}{d^{1/12}(\log n)^{5/6}}  \right) \leq \frac{C_\beta'}{\sqrt{d}}, $$ 
where the last inequality is by using the hypothesis that $d\geq (\log n)^{2}$. And choosing $R=C'''/d^{3/4}$, so that the assumption $R_*(\Sigma-I)^{1/2}\sigma_*(\Sigma-I)^{1/2}\leq R\leq\sigma_*(\Sigma-I)$ is indeed satisfied, we get $p_R\leq 1/d^{1/4}$ by Markov inequality. So applying the second estimate with $t=2\log n$, we additionally obtain that, with probability at least $1-1/n-1/d^{1/4}$, hence a fortiori at least $1-2/d^{1/4}$,
$$ \Vert \Sigma-I \Vert_\infty \leq \frac{C_\beta}{\sqrt{d}} \left( \frac{1}{\log n} + \frac{(\log n)^{3/4}}{\log n} + \frac{\log n}{d^{1/12}(\log n)^{2/3}} + \frac{(\log n)^{1/2}}{\log n} \right) \leq \frac{C_\beta'}{\sqrt{d}}, $$ 
where the last inequality is by using the hypothesis that $d\geq (\log n)^{4}$.
\end{proof}

\begin{theorem} \label{th:expander-TP}
Let $n, d\in \mathbb{N}$ with $d\geq(\log n)^{12}$ and consider $W_1,\ldots,W_d$ independent copies of an $n\times n$ random matrix $W$ satisfying Assumptions \ref{a:indep entries}, \ref{a: profil variance}, \ref{a: 4 moment}, \ref{a: max variance} and \ref{a:eta}. Set $\Sigma=\frac{1}{d}\sum_{s=1}^dW_s^*W_s$ and, for each $1\leq s\leq d$, $\tilde W_s=W_s\Sigma^{-1/2}$. Define the random CP map $\tilde\Phi$ on $\mathcal M_n(\mathbb C)$, having Kraus rank at most $d$, as
$$\tilde\Phi:X\in\mathcal M_n(\mathbb C)\mapsto \frac{1}{d}\sum_{s=1}^d \tilde W_sX\tilde W_s^*\in\mathcal M_n(\mathbb C). $$
Then, $\tilde\Phi$ is by construction TP, so in particular $s_1(\tilde\Phi) \geq 1$, and such that
 $$ \mathbb P\left( s_2(\tilde\Phi) \leq \frac{C_\beta}{\sqrt{d}} \right) \geq 1-\frac{6}{d^{1/4}}, $$
 where $C_\beta>0$ is a constant depending only on $\beta$.
\end{theorem}

\begin{proof}
    We will relate the second largest singular values of the CP map $\tilde\Phi$ to that of the CP map $\Phi$, from Theorem \ref{th:expander}. First observe that, by definition,
    $$ M_{\tilde\Phi} = M_\Phi\left(\Sigma^{-1/2}\otimes\overbar\Sigma^{-1/2}\right) . $$
    We thus have
    \begin{align*}
        \left\|M_{\tilde\Phi}-M_\Phi\right\|_\infty & = \left\|M_\Phi\left(\Sigma^{-1/2}\otimes\overbar\Sigma^{-1/2}-I\right)\right\|_\infty \\
        & \leq \left\|M_\Phi\right\|_\infty \left\|\Sigma^{-1/2}\otimes\overbar\Sigma^{-1/2}-I\right\|_\infty \\
        & \leq \left\|M_\Phi\right\|_\infty\left(\left\|I\right\|_\infty + \left\|\Sigma^{-1/2}\right\|_\infty\right) \left\|\Sigma^{-1/2}-I\right\|_\infty,
    \end{align*}
    where the last inequality follows from the triangle inequality, after writing 
    $$ \Sigma^{-1/2}\otimes\overbar\Sigma^{-1/2}-I=(\Sigma^{-1/2}-I)\otimes\overbar\Sigma^{-1/2}+I\otimes(\overbar\Sigma^{-1/2}-I). $$ 
    Now, we know from Theorem \ref{th:expander} that, with probability larger than $1-2/d^{1/4}$, $\|M_\Phi\|_\infty\leq C_\beta/\sqrt{d}$. And we know from Lemma \ref{lem:Sigma} that, with probability larger than $1-2/d^{1/4}$, $\|\Sigma^{-1/2}-I\|_\infty\leq C_\beta/\sqrt{d}$ and hence $\|\Sigma^{-1/2}\|_\infty\leq 1+C_\beta/\sqrt{d}$. We thus have that, with probability larger than $1-4/d^{1/4}$,
    $$ \left\|M_{\tilde\Phi}-M_\Phi\right\|_\infty \leq \frac{C_\beta'}{\sqrt{d}}. $$
    We can then conclude using Weyl inequalities for singular values that, with probability larger than $1-6/d^{1/4}$,
    $$ s_2\left(M_{\tilde\Phi}\right) \leq s_2\left(M_{\Phi}\right) + s_1\left(M_{\tilde\Phi}-M_{\Phi}\right) \leq \frac{C_\beta''}{\sqrt{d}}, $$
where the last inequality is because we know from Theorem \ref{th:expander} that, with probability larger than $1-2/d^{1/4}$, $s_2(M_{\Phi})\leq C_\beta/\sqrt{d}$. And this concludes the proof of Theorem \ref{th:expander-TP}.
\end{proof}

As a final comment, let us briefly explain how one could obtain, for the random CP map $\Phi$ described in Theorem \ref{th:expander}, a lower bound on its largest eigenvalue and an upper bound on its second largest eigenvalue (in addition to the lower bound on its largest singular value and the upper bound on its second largest singular value given in Theorem \ref{th:expander}). For that we need to make one extra assumption, namely that the variance matrix $\eta$ is doubly stochastic (and not only left stochastic and right sub-stochastic). This guarantees that the random CP map $\Phi$ satisfies on average both the TP condition \eqref{eq:TP-condition} and the unital condition \eqref{eq:unital-condition}, and more precisely that both $\Sigma=\frac{1}{d}\sum_{s=1}^dW_s^*W_s$ and $\Theta=\frac{1}{d}\sum_{s=1}^dW_sW_s^*$ are with high probability close to $I$ (by Lemma \ref{lem:Sigma} for $\Sigma$ and its exact analogue for $\Theta$). 

In order to lower bound the largest eigenvalue of $\Phi$, we use the following characterization (see e.g.~\cite[Theorem 6.3]{wolf2012}), valid for any irreducible CP map, which $\Phi$ is almost surely (by the same argument as in \cite[Fact 1.2]{lancien2022} and the discussion following it),
$$ |\lambda_1(\Phi)| = \sup_{X\geq 0} \sup\left\{\lambda\in\mathbb R : \Phi(X)\geq\lambda X\right\}. $$
We deduce from this identity that, in particular,
$$ |\lambda_1(\Phi)| \geq \sup\left\{\lambda\in\mathbb R : \Phi(I)\geq\lambda I\right\}. $$
Now, $\Phi(I)$ is nothing else than $\Theta$. So we know that, with high probability, $\|\Phi(I)-I\|_\infty\leq C_\beta/\sqrt{d}$, and consequently $\Phi(I)\geq(1-C_\beta/\sqrt{d})I$. We therefore have that, with high probability, $|\lambda_1(\Phi)|\geq 1-C_\beta/\sqrt{d}$.

In order to upper bound the second largest eigenvalue of $\Phi$, we start from Weyl's majorant theorem (see e.g.~\cite[Theorem~II.3.6]{MR1477662}), which tells us that
$$ |\lambda_1(\Phi)| + |\lambda_2(\Phi)| \leq s_1(\Phi) + s_2(\Phi). $$
Now, we know from Theorem \ref{th:expander} that, with high probability, $s_2(\Phi)\leq C_\beta'/\sqrt{d}$ and $s_1(\Phi)\leq 1+ C_\beta'/\sqrt{d}$ (by the exact same argument as the one used to show that $s_1(\Phi)\geq 1- C_\beta'/\sqrt{d}$). And we have just shown that, with high probability, $|\lambda_1(\Phi)|\geq 1-C_\beta/\sqrt{d}$. So we get that, with high probability, $|\lambda_2(\Phi)| \leq (C_\beta+2C_\beta')/\sqrt{d}$.

Note that, as we have just explained, when imposing that the variance matrix $\eta$ is doubly stochastic, we guarantee that our random CP map $\Phi$ is with high probability close to unital and thus has a fixed point $\hat\rho$ that is close to the maximally mixed state $I/n$, hence with close to maximal entropy. This means that, with this extra assumption on $\eta$, condition (3) in the definition of a quantum expander is automatically satisfied. However, our result remains interesting even without it, i.e.~whether the entropy of the fixed point $\hat\rho$ is large or small. Indeed, it provides examples of quantum channels that have a gap between their first and second largest singular values that is as large as it can be as soon as their number $d$ of Kraus operators is larger than $(\log n)^{12}$, hence in particular for $(\log n)^{12}\leq d\ll n$. And the latter regime is one where the construction is non-trivial whatever the entropy of $\hat\rho$ is, as explained in Section \ref{sec:def-expanders}.

\section{Summary and perspectives} \label{sec:conclusion}

This note aimed at shedding light on how recent results from random matrix theory can be readily used to exhibit quantum expanders from their classical counterparts. Let us in fact emphasize again on this point, which was already mentioned in the introduction. Our result provides a recipe for constructing a random optimal quantum expander from any optimal classical expander (either random or deterministic). Indeed, let $A\in\mathcal M_n(\mathbb R)$ be the adjacency matrix of a $d$-regular graph on $n$ vertices satisfying $d\geq (\log n)^{12}$ and $s_2(A)\leq C\sqrt{d}$, for some universal constant $C>0$. Let $W\in\mathcal M_n(\mathbb C)$ be a random matrix with independent centered entries, such that $\mathbb E|W_{ij}|^2=A_{ij}/d$, $1\leq i,j\leq n$. Suppose additionally that the moments of the $W_{ij}$'s satisfy Assumption \ref{a: 4 moment}. Next, sample $W_1,\ldots,W_d$ independent copies of $W$, set $K_s=W_s/\sqrt{d}$, $1\leq s\leq d$, and let $\Phi$ be the random CP map on $\mathcal M_n(\mathbb C)$ that has the $K_s$'s, suitably renormalized, as Kraus operators. Then $\Phi$ is a random quantum channel that satisfies with high probability $s_2(\Phi)\leq C'/\sqrt{d}$.

We have thus shown in this note that a wide class of random quantum channels are with high probability optimal quantum expanders. Prior to this work, only three specific models were known to have this property, all of them requiring a large amount of randomness to be implemented, as they relied on the sampling of either Haar unitary or Gaussian matrices. Here we prove that sparse matrices whose entries have a much simpler distribution (e.g.~Bernoulli random variables) actually do the job as well. This is a first step towards derandomization: exhaustive search for an explicit example of an optimal quantum expander might now be within reach. This could be of great interest in practice. Quantum expanders indeed find applications in many sub-fields of quantum information. The seminal ones were inspired by applications of classical expanders, and therefore more computer science orientated (quantum encryption \cite{ambainis2004}, quantum interactive proof systems \cite{BenAroya2008}, etc). But they also turned out to be crucial objects in quantum condensed matter physics, in particular to design many-body quantum systems where subsystems have large entanglement entropy but correlations between them that decay fast \cite{hastings2007,gonzalez2018,lancien2022}. 

One open question at this point is: how to generalize the results we have obtained here for random matrices $W_s$ with independent entries to the case of correlated entries? Indeed, the operator norm estimate from \cite{vanHandel,vanHandel2} that we use does not require independent entries. We only add this assumption to be able to get an optimal estimate without technical considerations. But the trade-off between correlation and moment control would be interesting to study. For instance, the same spectral gap scaling as the one we get is known to hold when the $W_s$'s are independent uniformly distributed unitaries \cite{hastings2007,pisier2014,timhadjelt2024}. One can wonder whether this setting (where all moments are tightly controlled) could be treated with our approach, and whether the result could be extended to the case of independent unitaries distributed according to more general measures, which are known to somehow resemble the uniform measure. Examples include approximate $t$-design measures (which are by definition such that their moments up to order $t$ are close to those of the uniform measure), the uniform measure on sub-groups of the unitary group (which have been shown to have moments that can be compared to those of the uniform measure in a systematic way \cite{bordenave2024}), etc. In fact, the case of $t$-design measures has been recently studied in \cite{lancien2024}, while this work was in its reviewing process. It has been proven there that sampling Kraus operators from a $2$-design measure is enough to obtain with high probability an optimal expander (and more generally that sampling Kraus operators from a $2t$-design measure is enough to obtain with high probability an optimal so-called $t$-copy tensor product expander).

Another interesting generalization would be to look at models of random quantum channels that have a particular symmetry, and thus a degenerate largest eigenvalue. Particularly relevant examples that have been studied in the past include the cases where Kraus operators are tensor power matrices \cite{hastings2009,harrow2009} or permutation matrices \cite{bordenave2019}. In physics, quantum channels whose Kraus operators are invariant under the action of a given group are ubiquitous, for instance when describing many-body quantum systems that have a local symmetry. It would be interesting to understand how the spectral gap of such channels typically scales, as this would give us information on how fast the correlations in such systems typically decay.

Finally, as mentioned in the introduction, we believe that the main mathematical result our analysis relies on (which is taken from \cite{vanHandel,vanHandel2} and stated here as Theorem \ref{t: van handel}) could find several other applications in quantum information. It could potentially be applied to any situation where one needs to estimate the deviation from average of a sum of independent random operators, and thus be useful to study all kinds of typical properties of quantum systems.

\subsection*{Acknowledgments} The authors are thankful to Ramon van Handel for bringing to their attention a bug in a first version of this work, and for the suggested corrections. They are also grateful to an anonymous reviewer whose extremely valuable comments helped improve this work greatly. Part of this work was conducted when the second named author visited the first named author in Institut Fourier at Universit\'e Grenoble Alpes. The first named author was supported by the ANR projects ESQuisses (grant number ANR-20-CE47-0014-01), STARS (grant number ANR-20-CE40-0008) and QTraj (grant number ANR-20-CE40-0024-01). The second named author was supported by the NYUAD Center for Interdisciplinary Data Science \& AI, funded by Tamkeen under the NYUAD Research Institute Award CG016.

\subsection*{Conflict of interest} On behalf of all authors, the corresponding author states that there is no conflict of interest.

\subsection*{Data availability} Data sharing not applicable to this article as no datasets were generated or analysed during the current study.

	\addcontentsline{toc}{section}{References}
	\bibliographystyle{alpha}
	\bibliography{references.bib}
	
\end{document}